\documentclass[11pt,oneside]{article}

\usepackage{palatino}
\usepackage[margin=1in]{geometry}
\usepackage{color,xcolor,soul}
\usepackage[abspage,user,savepos]{zref}
\usepackage[colorlinks,citecolor=blue]{hyperref}
\usepackage{natbib}
\usepackage{setspace}

\usepackage{amsmath,amssymb,amsfonts}
\usepackage{booktabs}
\usepackage[ruled]{algorithm2e}
\usepackage{graphicx}
\usepackage{textcomp}
\usepackage{multirow}
\usepackage{tikz}
\usepackage{subcaption}
\usepackage{tikz-3dplot}
\usepackage{tikz-network}

\usetikzlibrary{positioning,shapes,arrows,calc}
\usetikzlibrary{backgrounds}
\usetikzlibrary{arrows.meta,shapes.geometric}
\usetikzlibrary{quotes}

\def\BibTeX{{\rm B\kern-.05em{\sc i\kern-.025em b}\kern-.08em
    T\kern-.1667em\lower.7ex\hbox{E}\kern-.125emX}}

\tikzset{mynode/.style={draw, very thick, circle, minimum size=1cm}}
\tikzset{mynode/.style={draw, very thick, circle, minimum size=1cm},
    myarrow/.style={very thick, -Triangle}}

\ifx\BlackBox\undefined
\newcommand{\BlackBox}{\rule{1.5ex}{1.5ex}}
\fi

\ifx\QED\undefined
\def\QED{~\rule[-1pt]{5pt}{5pt}\par\medskip}
\fi

\ifx\proof\undefined
\newenvironment{proof}{\par\noindent{\it Proof\ }}{\hfill$\square$}
\fi

\ifx\theorem\undefined
\newtheorem{theorem}{Theorem}
\fi
\ifx\problem\undefined
\newtheorem{problem}{Problem}
\fi
\ifx\example\undefined
\newtheorem{example}[theorem]{Example}
\fi
\ifx\property\undefined
\newtheorem{property}{Property}
\fi
\ifx\lemma\undefined
\newtheorem{lemma}[theorem]{Lemma}
\fi
\ifx\proposition\undefined
\newtheorem{proposition}[theorem]{Proposition}
\fi
\ifx\remark\undefined
\newtheorem{remark}[theorem]{Remark}
\fi
\ifx\corollary\undefined
\newtheorem{corollary}[theorem]{Corollary}
\fi
\ifx\definition\undefined
\newtheorem{definition}[theorem]{Definition}
\fi
\ifx\conjecture\undefined
\newtheorem{conjecture}[theorem]{Conjecture}
\fi
\ifx\fact\undefined
\newtheorem{fact}[theorem]{Fact}
\fi
\ifx\claim\undefined
\newtheorem{claim}[theorem]{Claim}
\fi
\ifx\assumption\undefined
\newtheorem{assumption}[theorem]{Assumption}
\fi

\newcommand{\be}{\begin{equation}}
\newcommand{\ee}{\end{equation}}
\newcommand{\utau}{\underline{\tau}}
\newcommand{\otau}{\overline{\tau}}
\newcommand{\op}{\bar{p}}

\begin{document}
\date{}

\title{\LARGE \bf Controlling Traffic without Tolls: A Non-Monetary Framework for Autonomous Intersections}

\author{
Arda Kosay\footnotemark[1]\thanks{A. Kosay, Y. Saltan, and M. O. Sayin are with the Department of Electrical and Electronics Engineering, Bilkent University, Ankara 06800, T\"urkiye {\tt \small arda.kosay@bilkent.edu.tr, yusuf.saltan@bilkent.edu.tr, sayin@ee.bilkent.edu.tr}}
\and
Yusuf Saltan\footnotemark[1]
\and
Jyun-Jhe Wang\footnotemark[2]\thanks{J.-J. Wang and C.-W. Lin are with the Department of Computer Science and Information Engineering, National Taiwan University, Taipei 10617, Taiwan {\tt \small r13922195@ntu.edu.tw, cwlin@csie.ntu.edu.tw}}
\and
Chung-Wei Lin\footnotemark[2]
\and
Muhammed O. Sayin\footnotemark[1]
}

\maketitle
\thispagestyle{empty}

\bigskip

\begin{center}
\textbf{Abstract}
\end{center}
The increasing complexity of urban transportation systems, driven by connected and automated vehicles, calls for new modeling paradigms and scalable control strategies. We propose a non-monetary control framework that leverages autonomous intersection management to influence routing decisions without tolls. The approach uses timestamp-based scheduling adjustments at roadside units (RSUs) to introduce path-dependent delays or advancements, steering traffic toward socially efficient flows.
We develop a hierarchical architecture that separates real-time intersection control from network-level coordination. The resulting model admits a congestion-game formulation with path-dependent node costs. We establish the existence and essential uniqueness of equilibrium flows, eliminating ambiguities due to multiple equilibria and enabling a scalable and tractable bilevel optimization formulation for system-level incentive design. Experiments on the Sioux Falls network show that the proposed approach reduces the efficiency gap between user equilibrium and system-optimal flows by up to 71\% under realistic constraints.
These results demonstrate the potential of non-monetary, infrastructure-light control for next-generation intelligent transportation and urban mobility systems.

\bigskip


\bigskip
\bigskip

\newpage
\setcounter{page}{1}
\begin{spacing}{1.245}

\section{Introduction}
\label{sec:introduction}
Urban transportation systems are becoming increasingly complex with the integration of connected and automated vehicles and infrastructure. For example, in autonomous intersections, vehicles can communicate directly with roadside units (RSUs) that schedule each vehicle's passage \textit{individually} rather than on a per-lane basis \citep{ref:Dresner08,ref:Zhong20,ref:Khayatian20} (see Fig. \ref{fig:AIM}). 
This paradigm shift enables fine-grained, real-time control of intersection access at the vehicle level, creating new opportunities for coordinating traffic flows across the network. However, autonomous intersections introduce new computational and strategic challenges. RSUs must process a continuous stream of individual requests while the local scheduling decisions can influence drivers’ routing choices at a network-wide level—potentially leading to unintended congestion patterns. To turn these challenges into opportunities, our work addresses the following research question: \emph{How can intersection-level control be systematically designed to influence routing behavior and improve network-wide traffic efficiency without relying on monetary incentives?}

\begin{figure}[t!]
\centering
\includegraphics[width =\textwidth]{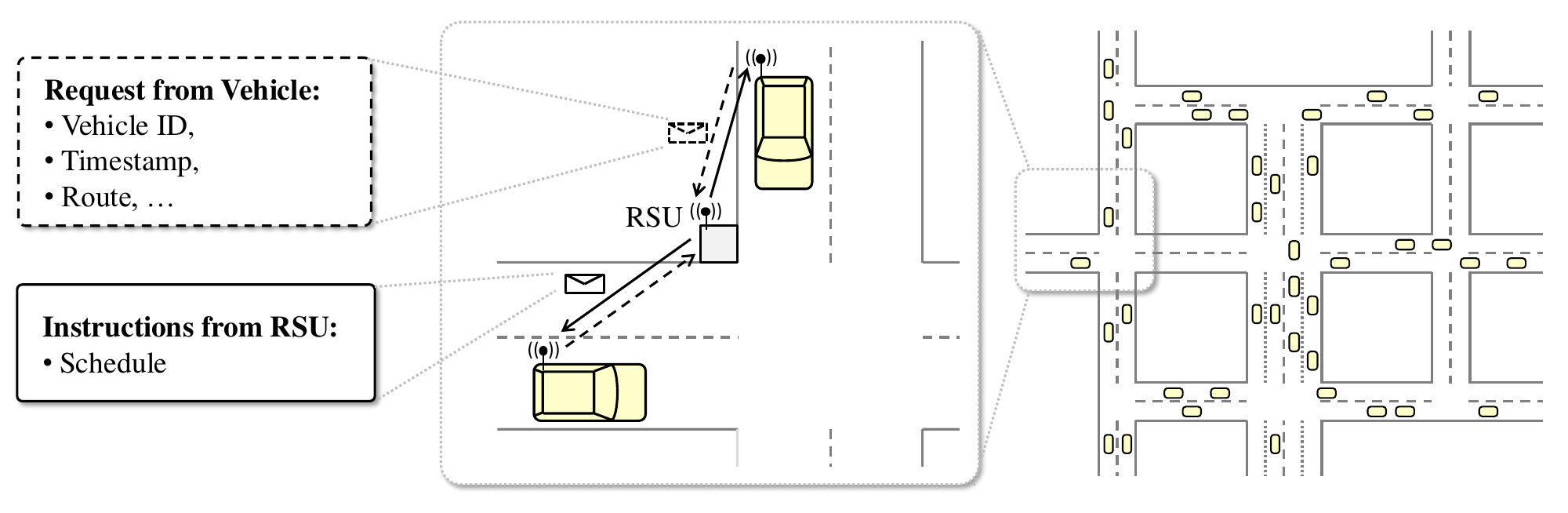}
\caption{Autonomous intersection management based on communication between vehicles and RSUs. }\label{fig:AIM}
\end{figure}

\subsection{Contributions and Approach}
We propose a new scalable and non-monetary intersection control framework with two key aspects: hierarchical design using timestamp modifications and offline bilevel optimization based on experimentally validated analytical model.

\textit{Hierarchical Design using Timestamp Modifications.} Coordinating multiple intersections to influence selfish drivers’ routing is computationally demanding due to the network’s large scale and the timing- and safety-critical nature of control. Fully centralized, end-to-end optimization is thus impractical and unreliable. To balance tractability and responsiveness, we adopt a hierarchical architecture that separates low-level control at individual intersections by RSUs from high-level planning across the network (see Fig. \ref{fig:model}).

Low-level control mechanisms, such as the widely-used first-come-first-serve (FCFS) protocols, generally ensure a positive correlation between \textit{earlier} requests and \textit{earlier} passage through the intersection to reduce the waiting time and they can ensure this correlation by using the timestamps of the driver requests in their decision-making.\footnote{Here, the term \textit{timestamp} refers to an abstract scheduling variable that may correspond to a vehicle’s request time, reported arrival time, or requested passage time, provided it remains positively correlated with the order in which the RSU grants intersection access.} We artificially \textit{delay or advance} these timestamps at the high-level planning. For example, when a timestamp is delayed, the RSU treats the driver request as if it has a later timestamp and correspondingly the driver is likely to experience longer waiting time. This can incentivize these drivers to choose other routes. The system planner determines these modifications by considering the entire traffic network and the counterfactual impacts of the decisions made, i.e., how drivers adapt their routing choices.

\textit{Offline Bilevel Optimization.} We consider a high-level planner using theoretical modeling of traffic flow across roadways and intersections to design timestamp modifications associated with the drivers' path choices in an offline manner. We approximate large driver populations via a nonatomic routing-game model with \textit{non-negative}, \textit{continuous}, and \textit{non-decreasing} latency functions. Since autonomous intersections are not yet widely deployed, we calibrate intersection costs using high-fidelity traffic simulations. These experiments show that timestamp modifications impact routes as \textit{additive}, path-specific costs. This experimentally grounded model ensures an \textit{essentially unique equilibrium flow} corresponding to the minimum of a new potential function we introduced. Hence, we can formulate a well-defined bilevel optimization problem solvable without the ambiguity of low-level outcome induced by equilibrium flows' multiplicity.

Our technical contributions are threefold:
\begin{itemize}
\item We propose a scalable hierarchical control architecture that separates local intersection scheduling from network-level coordination via non-monetary incentives.
\item We design a timestamp-based control mechanism that leverages RSU prioritization to induce additive cost adjustments along routes. This structure is critical for scalability: it guarantees essential uniqueness of equilibrium flows, avoids worst-case equilibrium selection, and enables a computationally tractable bilevel optimization framework for network-level control.
\item We experimentally validate the proposed framework through high-fidelity traffic simulations, demonstrating up to a 71\% reduction in the efficiency gap in the Sioux Falls network.
\end{itemize}
Importantly, this approach does not require additional infrastructure or monetary mechanisms, making it particularly attractive for practical deployment in emerging connected vehicle ecosystems.

\subsection{Related Work}

Existing works have studied self-interested routing decisions and intersection control largely \textit{in isolation}. In the following, we first review the literature on self-interested routing decisions and then on intersection management, highlighting the need to bridge these two research directions.

\textbf{Self-interested Routing Decisions.}
Traffic congestion and routing inefficiencies have long been central topics in transportation and game theory, with extensive research devoted to modeling and mitigating the resulting economic and environmental impacts.

\textit{Modeling Selfish Routing.}
Classical studies in nonatomic congestion games provide the foundational framework for understanding traffic flow in environments where users act selfishly. A central concept in this literature is \textit{Wardrop equilibrium} \citep{ref:Wardrop52}, which posits that no individual driver can unilaterally reduce their travel time by switching routes. Beckmann’s potential function \citep{ref:Beckmann56} offers a convex optimization formulation of these equilibria, enabling tractable analysis of flow distributions in networks.

However, equilibrium flows in routing games are typically inefficient compared to the socially optimal allocation—a discrepancy quantified by the \textit{Price of Anarchy} (PoA) \citep{ref:Koutsoupias99, ref:Roughgarden02}. For instance, in transportation networks modeled with the Bureau of Public Roads (BPR) latency function—often of quartic form—the PoA is bounded above by 2.151 \citep{ref:Roughgarden16}. Even marginal reductions in this inefficiency can translate into significant societal benefits, such as reduced commute times and economic productivity gains, as demonstrated in empirical studies like the analysis for Singapore in \citep[Appendix A.1]{ref:Monnot22}.

\begin{figure}[t!]
\centering
\includegraphics[width=\textwidth]{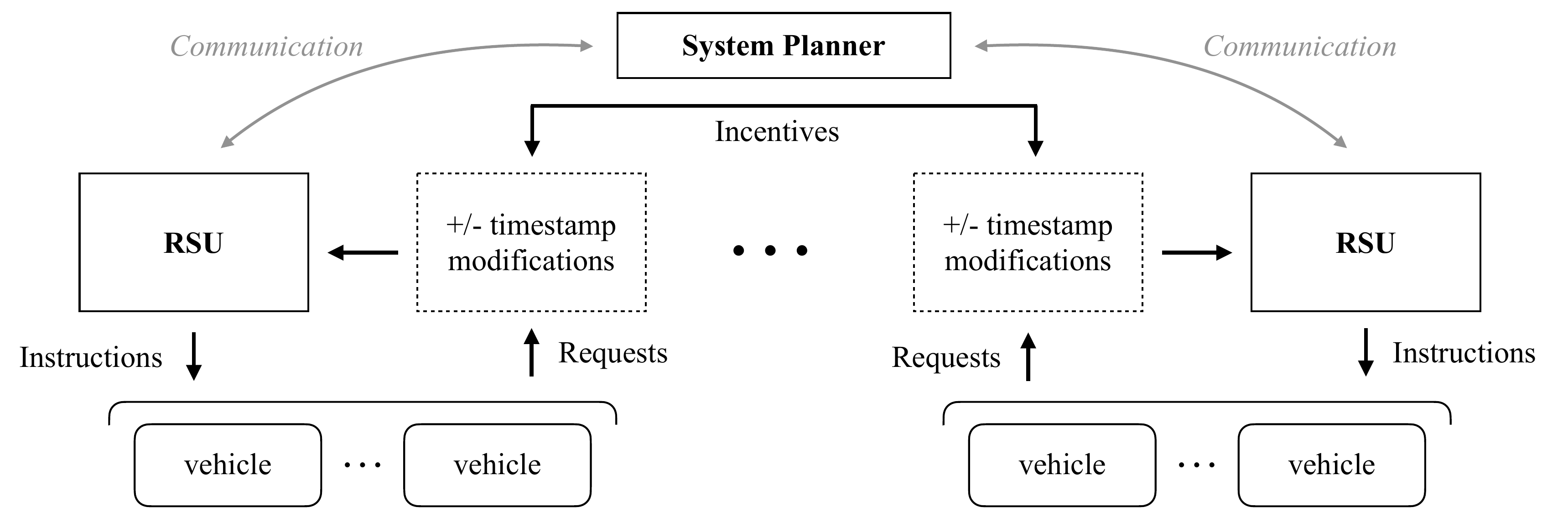}
\caption{Proposed two-layer solution separating incentive control and local intersection scheduling via timestamp modifications.}\label{fig:model}
\end{figure}

\textit{Mitigating Selfish Routing.}
Motivated by this inefficiency gap, a significant body of work has focused on aligning individual routing decisions with socially optimal outcomes through various incentive mechanisms. Monetary approaches—such as auction-based tolling and marginal cost pricing \citep{ref:Fleischer04, ref:Cole03, ref:Paccagnan21}—have demonstrated theoretical effectiveness but often raise equity concerns by disproportionately burdening less affluent drivers \citep{ref:Gemici18}. To address these equity issues, recent studies have investigated fairness-constrained traffic assignment problems \citep{ref:Angelelli21}, differentiated tolling schemes \citep{ref:Lazar21, ref:Maheshwari24}, and revenue redistribution mechanisms \citep{ref:Jalota24}. Alternatives like token-based systems \citep{ref:Sayin18, ref:Censi19} have also been proposed, but these rely on artificial currencies whose external value may be limited in practice.

In parallel, the role of information in routing games has been extensively investigated. The availability of route and congestion information can markedly affect traffic equilibria \citep{ref:Acemoglu18, ref:Wu21}. Studies on the informational Braess’ paradox reveal that incomplete information can sometimes improve equilibrium outcomes. Information design frameworks based on Bayesian persuasion \citep{ref:Kamenica11} have shown that carefully tailored private signals can outperform full disclosure. Nonetheless, these approaches are generally limited to simplified network models \citep{ref:Das17, ref:Cianfanelli23} and may be compromised by the presence of multiple competing information providers \citep{ref:Tavafoghi19}.

Beyond these monetary and information-based methods, efficiency improvements have also been pursued via capacity increases and network design. In networks with high-order nonlinear latency functions, even modest capacity enhancements can yield significant gains \citep{ref:Roughgarden07}. However, transportation networks are typically modeled using low-order latency functions—such as the BPR function—which means that generally substantial capacity increases would lead to meaningful improvements, rendering this approach less practical. Similarly, network design studies explore infrastructure modifications (e.g., adding or removing roads) to optimize traffic flow \citep{ref:Farahani13, ref:Cianfanelli23b}, which is a provably hard problem \citep{ref:Roughgarden06}.

\textbf{Intersection Management.}
Intersections are among the most critical bottlenecks in urban traffic networks. Traditional signal control methods \citep{ref:Papageorgiou03} rely on fixed or adaptive timing plans that often struggle to accommodate real-time fluctuations. Autonomous intersection management leverages vehicle-to-infrastructure communication to allow RSUs to schedule intersection usage on an individual request basis \citep{ref:Dresner08, ref:Lin19} (see Fig. \ref{fig:AIM}). To address the challenges of scheduling and prioritization in these systems, various approaches have been proposed. The simplest is the first-come-first-serve (FCFS) policy, which performs adequately under low traffic densities. \citet{ref:Dresner08} enhanced FCFS by introducing reservation tiles that allow for adjustable scheduling granularity. \citet{ref:Yang16} proposed a bilevel optimization model that integrates departure sequences with vehicle trajectory design based on arrival information, while \citet{ref:Lin19} developed a graph-based, divide-and-conquer scheduling policy that guarantees deadlock-free passage. More recently, reinforcement learning (RL) techniques have been applied to intersection management \citep{ref:Guan20, ref:Wu19, ref:Huang23}, reflecting the dynamic nature of the field. Furthermore, certain user groups, such as emergency services, may require preferential treatment—resulting in distinct travel times at intersections \citep{ref:Litman09, ref:Zhang15, ref:Harks18, ref:Hoefer11}. However, these solutions overlook the \textit{counterfactual impact} of intersection management on drivers' routing decisions—local improvements can inadvertently worsen overall congestion due to the strategic adaptations of drivers. Within the game-theoretical framework, \citet{ref:Scheffler22} examined edge-priority models in competitive packet routing games and \citet{ref:Saltan24} (a precursor to this work) demonstrated the benefits of strategic priority-based scheduling in the classical Pigou's example (with two parallel roads and one intersection).

\textbf{In this paper}, we model selfish routing via congestion games and introduce a \textit{strategic intersection control} framework as an additional layer on the existing intersection management solutions to delay or advance intersection usages based on route choices. This design allows intersection management to actively shape routing decisions, mitigating inefficiencies without monetary payments or infrastructural expansion. Our hierarchical architecture separates localized, real-time intersection control from system-wide incentive design for scalability and tractability, addressing general network topologies beyond prior illustrative examples in \citep{ref:Saltan24}.

The remainder of the paper is organized as follows. In Section \ref{sec:problem}, we introduce our new scalable and non-monetary intersection control framework. Section \ref{sec:main} formulates a bilevel optimization problem based on equilibrium flow characterization. In Section \ref{sec:examples}, we experimentally analyze and validate our framework through high-fidelity traffic simulations. Section \ref{sec:conclusion} concludes the paper and outlines potential directions for future research. Appendices \ref{app:potential} and \ref{app:existence} include the technical proofs and Appendix \ref{app:sumo-details} describes the experimental setup.

\section{Scalable Non-Monetary Control of Autonomous Intersections}\label{sec:problem}

Consider a traffic network with multiple autonomous intersections. Each intersection is equipped with a roadside unit (RSU) that collects requests from approaching vehicles. These requests include vehicle identification and timestamps (i.e., request times), enabling the system planner to track their trajectories from the initial terminal to the associated intersection (see Fig. \ref{fig:AIM}). These requests can also include the vehicles' routing plans such that the system planner can identify their entire trajectories. The RSUs schedule intersection usage based on these timestamps, using heuristic-based policies (e.g., first-come-first-serve \citep{ref:Dresner08, ref:Lin19}) or computationally intensive methods (e.g., planning and optimal control algorithms \citep{ref:Guan20, ref:Wu19, ref:Huang23}) while earlier requests generally lead to earlier passage through the intersection for shorter travel time.

We propose an additional \textit{layer of decision-making across multiple intersections} by a system planner to influence drivers' routing decisions dynamically. Specifically, based on a vehicle's trajectory, the system planner makes timestamp modifications—either delays (negative incentives) or advancements (positive incentives)—before requests are processed by RSUs, as illustrated in Fig. \ref{fig:model}. The RSUs execute their local scheduling policies (e.g., FCFS or optimal control) based on these adjusted timestamps. This mechanism aims to incentivize efficient routing decisions to mitigate congestion and improve system-wide travel quality. In the following, we formalize these timestamp offsets as additive, path-dependent node costs.

\subsection{Analytical Framework}
We can model the traffic network as directed graphs $G=(V,E)$, where $V$ represents autonomous intersections and $E$ represents roads. We consider a single-commodity network with a unit-mass of traffic flowing from a source terminal $s\in V$ to a destination terminal $t\in V$, for notational simplicity.\footnote{The analytical framework and the potential-based equilibrium results extend in a routine manner to finitely many commodities by introducing commodity-specific path sets and flow-conservation constraints (e.g. see \citep[Section~6.4]{ref:Acemoglu18}).} A \textit{path} $p\in\mathcal{P}$ is a sequence of nodes and edges from $s$ to $t$, with $\mathcal{P}$ denoting the \textit{finite} set of such paths. The traffic \textit{flow} is represented as $f=(f_p)_{p\in \mathcal{P}}$, where $f_p$ denotes the fraction of traffic using path $p$. We assume an \textit{inelastic} demand such that the \textit{feasible} flow $f$ satisfies the flow conservation constraints: 
\be\label{eq:feasible}
\mathcal{F} := \left\{f=(f_p)_{p\in \mathcal{P}} : \sum_{p\in\mathcal{P}} f_p = 1,\mbox{ and } f_p\geq 0\;\forall p\right\},
\ee
i.e., $f\in \mathcal{F}$.
For each edge $e\in E$ and node $v\in V$, we define the edge flow $f_e$ and node flow $f_v$ as
\begin{align}
f_e := \sum_{p\in \mathcal{P}:e\in p} f_p \quad \mbox{and} \quad f_v := \sum_{p\in \mathcal{P}:v\in p} f_p.
\end{align}

Each edge $e$ has an edge cost function $c_e:[0,1]\rightarrow[0,\infty)$, mapping flow $f_e\geq 0$ to an average cost, e.g., corresponding to the expected travel time. Unlike classical routing games, at each intersection $v$, we also consider \textit{node cost} $c_v^p(f_v, u_v^p)$, corresponding to the expected intersection-usage time. Here, $u_v^p\in[-\utau,\otau]$ for some $\utau,\otau \geq 0$ is the path-specific incentive, modeling the impact of timestamp modifications on the intersection cost. We model the impact of timestamp modifications as an \textit{additive} incentive term $u_v^p$ on a base delay based on the experiments in high-fidelity traffic simulations, discussed in Section \ref{sec:examples}. Formally, the cost of a node $v$ for drivers choosing path $p$ is given by
\be\label{eq:nodecost}
c_v^p(f_v, u_v^p) := c_v(f_v) + u_v^p,
\ee
where $c_v:[0,1]\rightarrow[\utau,\infty)$ is the \textit{base} node cost, capturing the intrinsic (average) intersection delay induced by the intersection topology and how the associated RSU schedules the intersection usage.\footnote{The lower bound $\utau \geq 0$ on $c_v(\cdot)$ and $u_v^p$ ensures that node cost $c_v^p(\cdot)\geq 0$.} The term $u_v^p$, determined by the system planner, \textit{increases or decreases} intersection cost so that we can shape route selections.

The following assumption ensures well-behaved costs.

\begin{assumption}\label{assm:cost}
Costs $c_e(\cdot)$ and $c_v^p(\cdot)$ are \textit{non-negative}, \textit{continuous}, and \textit{non-decreasing} in flow $f_e$ and $f_v$, respectively.
\end{assumption}

\begin{remark}[Average/Expected Cost]
These costs functions correspond to the average/expected cost, abstracting away from the occasional differences that may occur in practice. For example,
for vehicles approaching an intersection at the same lane, the timestamp modifications can lead to scenarios where a vehicle might have an earlier timestamp than the vehicles in the front. RSUs can determine how to resolve such inconsistencies and this may result in (operationally) delayed or expedited intersection usage for some vehicles. However, locations and route choices of vehicles in an approaching lane are random and such impacts average out across time.
\end{remark}

The cost of a path $p\in \mathcal{P}$ is given by
\be\label{eq:pathcost}
c_p(f, u) := \sum_{e\in p} c_e(f_e) + \sum_{v\in p}c_v^p(f_v,u_v^p),
\ee
where $u= (u_v^p\in [-\utau,\otau])_{v\in V, p\in\mathcal{P}}$.
Self-interested drivers choose paths that minimize their individual travel costs, leading to a \textit{Wardrop equilibrium}.

\begin{definition}[Wardrop Equilibrium]\label{def:equilibrium}
For a given incentive profile $u$, a feasible flow $f=(f_p)_{p\in\mathcal{P}}$ is an equilibrium flow if and only if we have
$c_p(f, u) \leq c_{\op}(f, u)$
for all paths $\op,p\in \mathcal{P}$ such that $f_p>0$.
\end{definition}

\begin{remark}[Equilibrium Behavior]
The drivers know or learn the path costs under the incentives determined by the system planner. They look for the paths with the smallest cost and \textit{stabilize} at Wardrop equilibrium in which none has incentive to change their routes. This equilibrium captures the \textit{steady‐state} outcome of route adjustments over time under fixed incentive schemes, rather than repeated real‐time re‐optimization by each driver and the system planner. In other words, the user equilibrium represents the time‐averaged distribution of selfish route choices under the designed incentives.
\end{remark}

For a given incentive profile $u$, the \textit{social cost} of a flow $f$ is defined as
\be\label{eq:socialcost}
C(f, u) := \sum_{p\in \mathcal{P}} c_p(f, u) \cdot f_p.
\ee
We can separate path-invariant costs from path-specific incentives. Substituting \eqref{eq:pathcost} into \eqref{eq:socialcost} and changing the order of summation, we can rewrite the social cost as 
\begin{flalign}
C(f, u)= C^o(f) + \sum_{p\in\mathcal{P}} f_p \cdot u_p, 
\end{flalign}
where 
\be
C^o(f) := \sum_{e\in E} c_e(f_e) f_e + \sum_{v\in V} c_v(f_v) f_v
\ee
is the cost without timestamp modifications and $u_p := \sum_{v\in p}u_v^p$ is the total incentive adjustment along path $p$.

\begin{figure}[t!]
    \centering
    \includegraphics[width=0.8\linewidth]{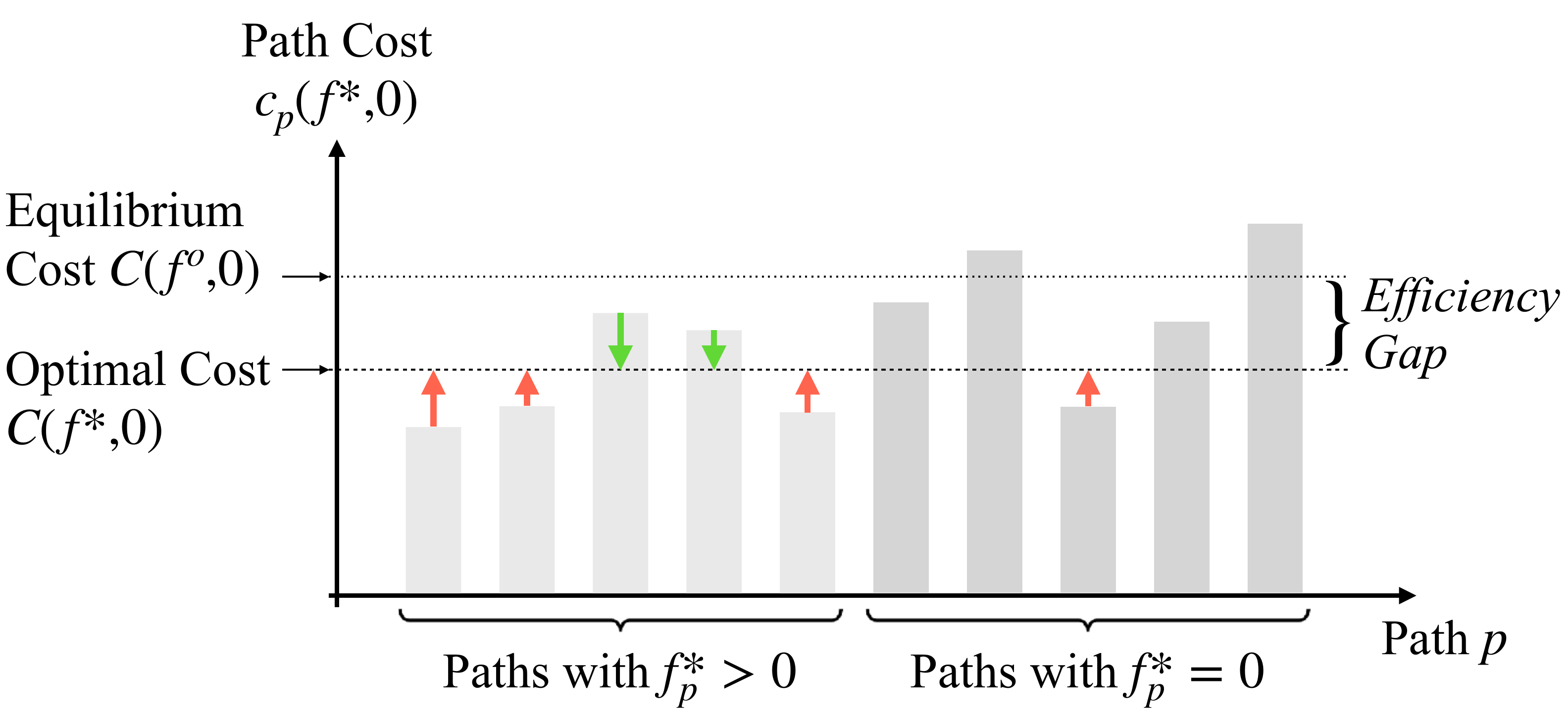}
    \caption{A figurative plot of path costs vs paths under optimal flow $f^*$ without any incentives applied. Vertical arrows depict how incentives can increase or decrease path costs so that optimal flow can become an equilibrium flow under incentives while still attaining the optimal cost.}
    \label{fig:optimal}
\end{figure}

Without incentives, i.e., $u_v^p = 0$ for all $v,p$, flows minimizing the social cost are not necessarily equilibrium flows since paths with positive flows can have costs higher than other paths, as illustrated in Fig. \ref{fig:optimal}. Correspondingly, there can be an \textit{efficiency gap} between social costs under equilibrium flows and the minimum cost. 
Our goal is to design incentives $u$ to improve equilibrium efficiency in terms of the social cost \eqref{eq:socialcost}. For example, we can choose $u_p$'s such that all paths with positive flows under the optimal flow attain the optimal cost as the path cost (so that their weighted average \eqref{eq:socialcost} is the minimum cost) while all other paths have higher or equal costs, as depicted via arrows in Fig. \ref{fig:optimal}. However, such solutions may not be applicable in general due to the constraints on the incentives. Therefore, we need a principled optimization framework to balance the trade-off between equilibrium and efficiency with systematic guarantees.

\begin{remark}[Constraints]
We consider bounded incentives $u_v^p\in [-\underline{\tau},\overline{\tau}]$ for all $v,p$. There can be additional constraints in practical implementations. For example, the system planner (communicating with RSUs) can track the sequence of edges and nodes used by each vehicle based on their requests. Thus, the planner can always identify the drivers' paths yet \textit{partially} since paths overlapping up to the current intersection and differing afterwards are indistinguishable at the current intersection. Moreover, the planner can even identify entire paths of drivers if vehicles are (truthfully) reporting their routing plans in their requests.
\end{remark}

Let $\mathcal{U} \subset \prod_{(v,p)\in V\times \mathcal{P}} [-\utau,\otau]$ be the feasible incentive set, capturing practical constraints. In particular, $\mathcal{U}$ need not require RSUs to know each driver's full route or future path. For example, $\mathcal{U}$ may restrict incentives to depend only on currently observed path information, so that vehicles with the same observed incoming path receive the same incentive regardless of their downstream routes. Then, the following is a formal description of the optimal incentive design problem.

\begin{problem}[Optimal Incentive Design]\label{prob:main}
Solve
\be\label{eq:prob_main}
\min_{u\in\mathcal{U}}\max_{f\in \mathcal{F}} \left\{C^o(f) + f\cdot u\right\}
\ee
subject to the equilibrium condition 
$c_p(f,u)\leq c_{\op}(f,u)$ for all $p,\op\in\mathcal{P}$ and $f_p>0$,
where $f\cdot u := \sum_p f_p\cdot u_p$ denotes the total additive cost induced by the incentives.
\end{problem}

Given incentives, there might be multiple equilibrium flows. Problem \ref{prob:main} seeks the \textit{worst-case} equilibrium improvement, to address the multiplicity issue. Later, we show the essential uniqueness of equilibrium flows to simplify the problem.

\section{Bi-level Optimization Formulation}\label{sec:main}
    
 Our goal is to design incentives in such a way that the induced equilibrium flow minimizes the overall social cost. However, the twist is that the equilibrium flow does not arise arbitrarily; it is the outcome of the drivers' selfish behavior. In other words, for any given incentive structure, drivers adjust their routing to minimize their individual perceived costs. Therefore, we first focus on characterizing the drivers' equilibrium behavior in our model.

In many classical routing models, edges bear the primary costs while nodes are typically treated as auxiliary, cost‐free elements, e.g., see \citep{ref:Roughgarden07}. In contrast, our timestamp modifications yield \textbf{additive path‐specific costs} that directly affect drivers’ routing decisions. This raises a natural question: \textit{Can nodes with such path‐dependent costs be analyzed by directly treating them as edges?}

\begin{remark}
Although intersection costs are path-dependent, they do not depend on driver identity. All drivers selecting the same path incur the same cost, preserving the structure of classical congestion games. This distinction is important, as more general prioritization schemes may break this structure and lead to loss of desirable properties such as existence and essential uniqueness of equilibrium flows.
\end{remark}

Another technical challenge is that different paths experience different node costs that depend on the overall node flow (rather than solely on the flow along that path). Such a distinction prevents a straightforward aggregation of node and edge costs as is standard in congestion models. Our network model with \textit{additive} path‐specific node costs can be transformed into a standard routing model by interpreting nodes as a specific type of hyper‐edges. Consider the following example:

\begin{figure}[t]
    \centering
    \begin{subfigure}{.5\textwidth}
    \centering
    \begin{tikzpicture}[scale=0.65, >=stealth, node distance=1.7cm, font=\small]
    
       \node[draw, circle,red] (s) at (0,0) {$s$};
       \node[draw, circle,red] (v) at (3,0) {$v$};
       \node[draw, circle,red] (t) at (6,0) {$t$};
    
       \draw[->, thick,red] (s) edge[bend left=25] node[above]{$c_{e_1}(f_{e_1})$} (v);
       \draw[->, thick] (s) edge[bend right=25] node[below]{$c_{e_2}(f_{e_2})$} (v);
    
       \draw[->, thick,red] (v) edge[bend left=25] node[above]{$c_{e_3}(f_{e_3})$} (t);
       \draw[->, thick] (v) edge[bend right=25] node[below]{$c_{e_4}(f_{e_4})$} (t);
    
    \end{tikzpicture}
    \caption{\small Original network}\label{fig:node_as_edge_a}
    \end{subfigure}

    \vspace{.1in}
    \begin{subfigure}{.5\textwidth}
    \centering
    \begin{tikzpicture}[scale=0.65, >=stealth, node distance=1.7cm, font=\small]

       \node[draw, circle,red] (s2) at (-1.3,0) {$s$};
       \node[draw, circle,dashed,red] (nIn) at (2.3,0) {$v_0$};
       \node[draw, circle,dashed,red] (nOut) at (4.8,0) {$v_1$};
       \node[draw, circle,dashed,red] (cNode) at (7.3,0) {$v_2$};
       \node[draw, circle,red] (t2) at (10.7,0) {$t$};
    
       \draw[->, thick,red] (s2) edge[bend left=25] node[above]{$c_{e_1}(f_{e_1})$} (nIn);
       \draw[->, thick] (s2) edge[bend right=25] node[below]{$c_{e_2}(f_{e_2})$} (nIn);
    
       \draw[->, very thick,dashed,red] (nIn) -- node[above]{$c_v(f_v)$} (nOut);
    
       \draw[->, thick,dashed,red] (nOut) to[out=30, in=150]
           node[above,pos=0.5]{$u_v^{1}$} (cNode);
       \draw[->, thick,dashed] (nOut) to[out=10, in=170]
           node[above,pos=0.6]{} (cNode);
       \draw[->, thick,dashed] (nOut) to[out=-10, in=190]
           node[below,pos=0.6]{} (cNode);
       \draw[->, thick,dashed] (nOut) to[out=-30, in=210]
           node[below,pos=0.5]{$u_v^4$} (cNode);
    
       \draw[->, thick,red] (cNode) edge[bend left=25] node[above]{$c_{e_3}(f_{e_3})$} (t2);
       \draw[->, thick] (cNode) edge[bend right=25] node[below]{$c_{e_4}(f_{e_4})$} (t2);
    
       \draw[thick, rounded corners=0.4cm] 
         ($(nIn)+(-0.8,1.4)$) 
         rectangle 
         ($(cNode)+(0.8,-1.4)$);
    
    \end{tikzpicture}
    \caption{Transformed Network}\label{fig:node_as_edge_b}
    \end{subfigure}
    \caption{Comparison of the original (a) three-node/four-edge network and 
    (b) the transformed version with node splitting. The path $p_1=se_1ve_3t$ and its counterpart in the transformed network are highlighted in red.}
    \label{fig:node_as_edge}
    \end{figure}

\begin{example}\label{example}
Consider a network with three nodes $s,v,t$ and four edges $e_1,\ldots,e_4$, as depicted in Fig. \ref{fig:node_as_edge_a}. Here, there are four paths, e.g., $\mathcal{P}=\{p_1=se_1ve_3t, p_2=se_1ve_4t, p_3=se_2ve_3t, p_4=se_2ve_4t\}$. Each path incurs a different node cost. We can transform the network by splitting node $v$ into two components (See Fig. \ref{fig:node_as_edge_b}). The first component is an edge through which all node flow passes, capturing the baseline node cost $c_v(f_v)$. The second component is a collection of path-specific parallel edges that allow only the corresponding path flow to pass through. These path-specific edges incur constant costs that do not depend on the node flow, thereby fitting the transformed model into the classical routing framework. 
\end{example}

This transformation implies that we can apply a rich array of traditional routing techniques to networks with additive path‐specific costs, thereby bridging the gap between classical and path‐specific models.

\subsection{Equilibrium Characterization}

A key property of congestion games is that equilibrium flows can often be characterized as solutions of an optimization problem. To analyze equilibrium behavior in our setting, we establish a \textit{potential function} whose minimizers correspond exactly to equilibrium flows.

In classical congestion models with edge flows $f_e$ and cost functions $c_e(\cdot)$, Beckmann’s potential function is given by
\[
\Phi_{\text{Beckmann}}(f) = \sum_{e\in E} \int_0^{f_e} c_e(z)\,dz ,
\]
and its minimizers coincide with equilibrium flows \citep{ref:Beckmann56}.
Based on the hyper-edge transformation, we extend this classical construction to account for intersection delays and timestamp-based incentives and define
\be\label{eq:newpotential}
\Phi(f,u) := \sum_{e\in E} \int_0^{f_e}c_e(x)dx + \sum_{v\in V}\int_0^{f_v}c_v(y)dy + f\cdot u.
\ee
This defines a potential function since under a fixed incentive profile $u$, the marginal variation of $\Phi(f,u)$ with respect to a path flow matches the corresponding perceived path cost. The following proposition characterizes equilibrium flows through the potential function $\Phi(f,u)$ for given $u$.

\begin{proposition}[Potential Function]\label{prop:potential}
Consider a routing game with additive path-specific costs and suppose that Assumption \ref{assm:cost} (non-negative, continuous, and non-decreasing costs) holds. Then, for a given incentive profile $u\in \mathcal{U}$, a flow $f=(f_p)_{p\in\mathcal{P}}$ is an \textbf{equilibrium flow} if and only if 
\be\label{eq:potential_min}
f\in \arg\min_{\tilde{f}\in \mathcal{F}} \;\Phi(\tilde{f},u).
\ee
\end{proposition}

The proof (see Appendix \ref{app:potential}) extends the classical characterization of Wardrop equilibrium \citep{ref:Smith79} to accommodate the additive path‐specific costs.

Our next result establishes existence and uniqueness properties of equilibrium flows.

\begin{theorem}[Existence and Uniqueness] \label{thm:existence}
Consider a routing game with additive path-specific costs under a given exogenous incentive profile $u$. If Assumption \ref{assm:cost} holds, then
\begin{itemize}
\item \textbf{Existence:} An equilibrium flow always exists.
\item \textbf{Essential Uniqueness:} For any two equilibrium flows $f$ and $\bar{f}$, the social cost is identical, i.e., $C(f, u) = C(\bar{f}, u)$.
\end{itemize}
\end{theorem}

The proof (see Appendix \ref{app:existence}) follows from an extension of standard results in classical routing games. The essential uniqueness of equilibrium flows implies that from the worst-case to the best-case, all equilibrium improvements lead to the same problem formulation in Problem \ref{prob:main}.
  
\subsection{A Bilevel Optimization Problem}

Proposition \ref{prop:potential} yields that in the absence of incentives, the drivers' decision-making is captured by the potential function:
\be\label{eq:basepotential}
\Phi^o(f) := \sum_{e\in E} \int_0^{f_e}c_e(x)dx + \sum_{v\in V}\int_0^{f_v}c_v(y)dy.
\ee
In other words, the corresponding equilibrium flows are the minimizer of $\Phi^o(f)$. When we add the incentives $u$, the equilibrium flows become the minimizer of $\Phi(f,u) = \Phi^o(f) + f\cdot u$ and all lead to the same social cost. 
Based on this observation, Problem \ref{prob:main} can be cast as the following bilevel optimization problem.

\begin{problem}[Bilevel Optimization Problem] \label{prob:bilevel}
Solve
\be\label{eq:SE}
\min_{u\in \mathcal{U}}\{C^o(f(u)) + f(u)\cdot u\},
\ee
subject to the equilibrium condition
\be\label{eq:SE_constraint}
f(u) \in \arg\min_{\tilde{f}\in\mathcal{F}} \{\Phi^o(\tilde{f}) + \tilde{f}\cdot u\}.
\ee
\end{problem}

\begin{remark}
In the absence of essential uniqueness (Theorem~\ref{thm:existence}), the bilevel optimization problem becomes ambiguous due to the presence of multiple equilibria with potentially different social costs. In such cases, the upper-level objective depends on (worst-case) equilibrium selection, which significantly complicates computation and limits scalability.
\end{remark}

At the lower level, for any fixed incentive vector $u$, drivers adjust their route choices to minimize their perceived cost, leading to a Wardrop equilibrium characterized by the minimization of $\Phi^o(f) + f \cdot u$. At the upper level, a central planner selects $u$ to minimize the overall social cost, taking into account that the resulting flow $f(u)$ arises as the equilibrium outcome of the lower-level problem. This interaction naturally yields a bilevel optimization structure capturing the coupling between system-level control and user behavior.

\section{Experimental Analyses}\label{sec:examples}

In this section, we demonstrate the effectiveness of our proposed timestamp modification scheme on Braess network and Sioux Falls transportation network. We further experimentally validate the additive path-dependent cost model.

\subsection{Braess Network}

The Braess network shown in Fig. \ref{fig:notseries} features three paths: $\mathcal{P}=\{p_1=se_1ve_3t, p_2=se_1ve_5we_4t, p_3=se_2we_4t \}$. We consider quadratic and quartic cost models for the edges $e_1$ and $e_4$, defined as follows: 
\begin{itemize}
\item Quadratic: $c(f) = f-0.5 f^2$,
\item Quartic: $c(f) = 0.5 f - 0.5 f^2 - f^3 + 2 f^4$.
\end{itemize}
Base node costs are $c_v(f_v) = f_v$ and $c_{w}(f_w) = f_w$. Timestamp modifications are constrained to $u_v^p,u_w^p\in [0,0.2]$, allowing delays \textit{only}. To solve the bilevel optimization problem (Problem \ref{prob:bilevel}), we employed several global optimization algorithms: simplicial homology global optimization algorithm (SHGO) \citep{ref:Endres18}, differential evolution (DE) \citep{ref:Storn97}, and dual annealing (DA) \citep{ref:Xiang97}. All algorithms converged to the optimal incentive structure $(u_v^{p_1},u_v^{p_2},u_w^{p_2},u_w^{p_3}) = (0,0.2,0.2,0)$, effectively delaying only path $p_2$, as illustrated in Fig. \ref{fig:optimal}.

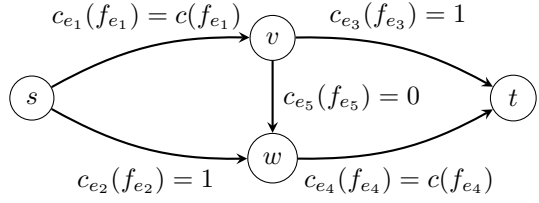
\begin{figure}[t!]
    \centering
\begin{tikzpicture}[scale=0.8, >=stealth, node distance=1.7cm, font=\small]
\node[draw, circle](s) at (0,0){$s$};
\node[draw, circle](v) at (4,1){$v$};
\node[draw, circle](w) at (4,-1){$w$};
\node[draw, circle](t) at (8,0){$t$};
\draw[->, thick](s) edge[bend left=15] node[midway, above=.1]{$c_{e_1}(f_{e_1}) = c(f_{e_1})$} (v);
\draw[->, thick](s) edge[bend right=15] node[midway, below=.1]{$c_{e_2}(f_{e_2}) = 1$} (w);
\draw[->, thick](v) edge[bend left=15] node[midway, above=.1]{$c_{e_3}(f_{e_3}) = 1$} (t);
\draw[->, thick](w) edge[bend right=15] node[midway, below=.1]{$c_{e_4}(f_{e_4}) = c(f_{e_4})$} (t);
\draw[->, thick](v)--(w) node[midway, right]{$c_{e_5}(f_{e_5}) = 0$};
\end{tikzpicture}
\caption{A Braess network.}\label{fig:notseries}
\label{graph:braess}
\end{figure}

\begin{table}[t!]
\centering
 \renewcommand{\arraystretch}{1.5}
 \caption{Flows $(f_{p_1},f_{p_2},f_{p_3})$ and social costs with and without incentives for the Braess network.}\label{tab:results}
\begin{tabular}{p{5em}|p{7.5em}|p{2em}|p{7.5em}|p{2em}}
\multirow{2}{3em}{\bf Cost Model} & \multicolumn{2}{|p{9.5em}|}{\bf Quadratic} &  \multicolumn{2}{|p{9.5em}}{\bf Quartic} \\\cline{2-5}
 & Flow  & Cost & Flow & Cost \\\hline\hline
Equilibrium without Incentives & $(0.41,0.17,0.41)$ & $2$ & $(0.34,0.31,0.34)$ & $2$   \\\hline
Optimal without Incentives & $(0.5,0,0.5)$ & $1.87$ & $(0.5,0,0.5)$ & $1.69$  \\\hline
Equilibrium with Incentives & $(0.5,0,0.5)$ & $1.87$ & $(0.5,0,0.5)$ & $1.69$
\end{tabular}
\end{table}

Table~\ref{tab:results} summarizes the equilibrium and socially optimal flows with and without incentives for each cost model. The results demonstrate optimal flow and cost under equilibrium when incentives are applied.

\subsection{Experimental Analyses on the Sioux Falls Network}
\label{subsec:sioux-falls-experiment}

We next focus on extensive experimental analyses in realistic microscopic traffic simulations.  
To this end, we use the Sioux Falls transportation network (24 nodes, 76 links, 552 OD pairs), a well-established benchmark in transportation engineering and algorithmic game theory, e.g., see \citep{ref:LeBlanc75}.  
Fig.~\ref{fig:sioux-overview} shows the network layout.  
The benchmark provides a realistic medium-scale topology with empirically calibrated BPR edge cost functions.  
However, the dataset in \citet{ref:LeBlanc75} does not include intersection (node) delays.  
To address this gap and ground the theoretical model in realistic data, we construct intersection-level simulations using the Eclipse \textbf{S}imulation of \textbf{U}rban \textbf{MO}bility (SUMO) platform \citep{ref:Lopez18}—an open-source microscopic simulator widely used to model vehicle dynamics, signal control, and queuing effects.

\begin{figure}[t!]
  \centering
  \begin{subfigure}[b]{0.48\textwidth}
    \centering
    \includegraphics[height=2.5in]{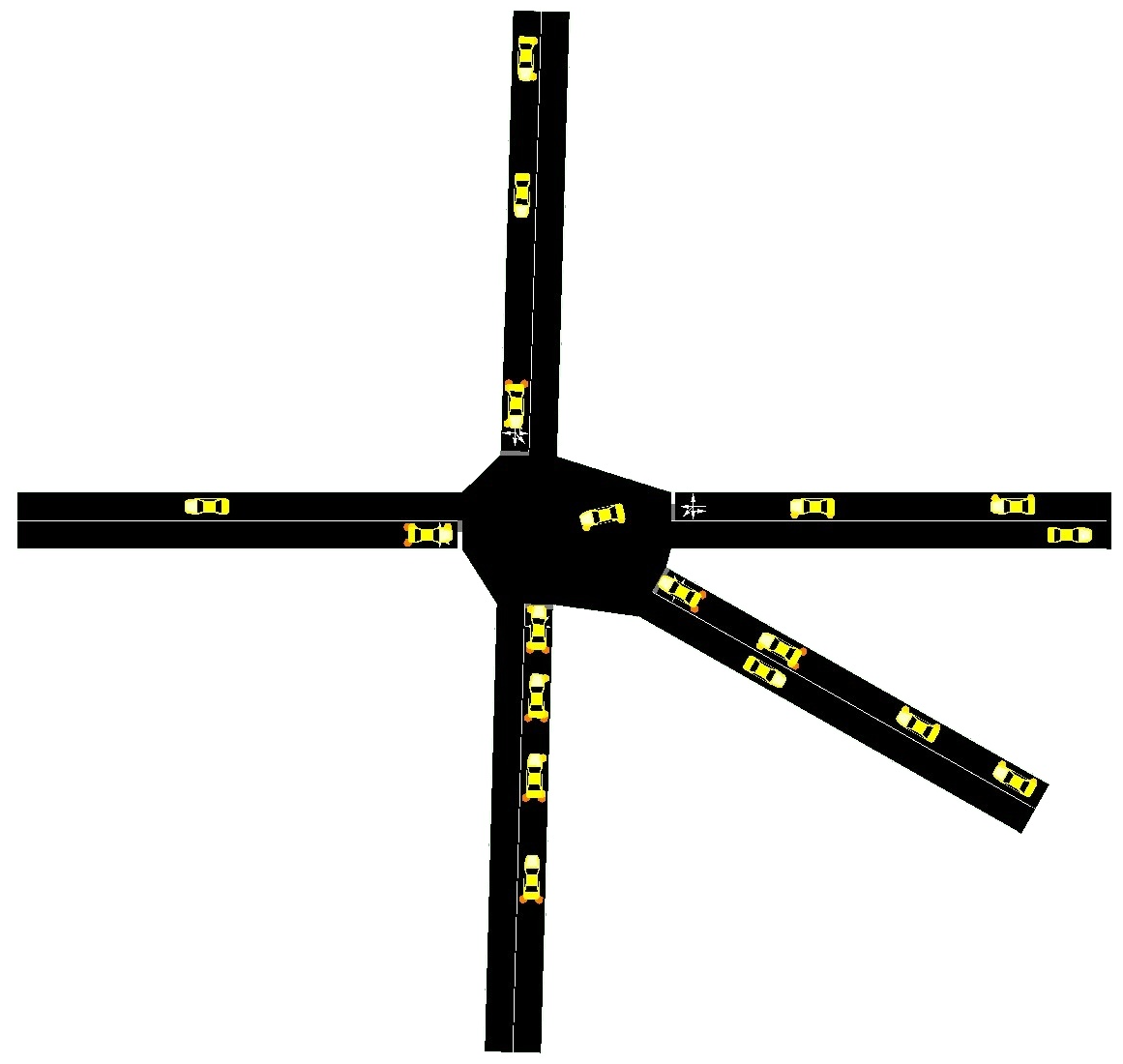}
    \caption{{SUMO snapshot for Node 10 (25\,m control zone).}}
    \label{fig:sumo-node10}
  \end{subfigure}\quad
  \begin{subfigure}[b]{0.48\textwidth}
    \centering \includegraphics[height=2.5in]{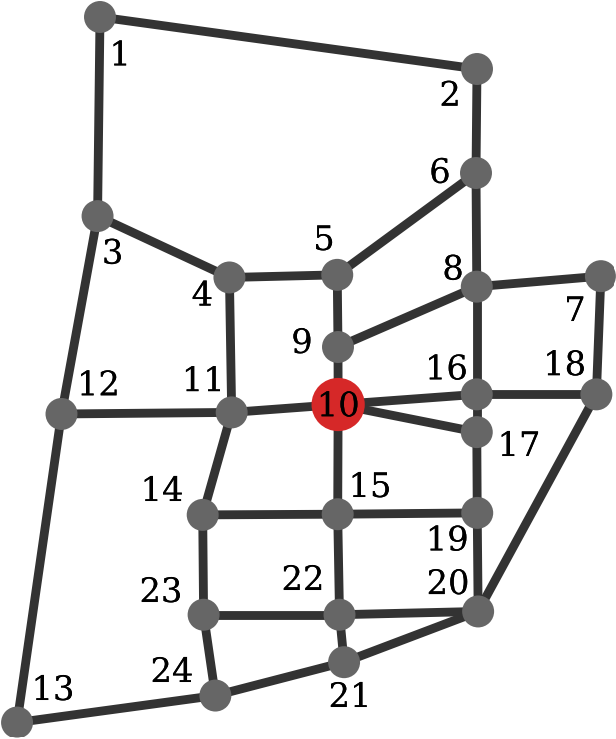}
    \caption{{Sioux Falls network (24 nodes, 76 links).}}
    \label{fig:sioux-overview}
  \end{subfigure}
  \caption{{Microscopic–macroscopic linkage: (a)  example intersection simulated in the SUMO, (b) full network.}}
  \label{fig:sumo-network-fit-triad}
\end{figure}

Each of the 24 intersections is modeled as an independent SUMO scenario replicating its geometry, lane configuration, and signal policy.  
Vehicle departures follow a Poisson process, introducing stochasticity that captures natural variations in vehicle arrivals.  
For each run, we measure vehicle travel times from a point 25\,m upstream of the stop line—corresponding to the control zone of RSU—until the vehicle fully exits the intersection.  
This measurement captures both queuing and traversal delays.  
Averaging across randomized rollouts yields smooth and repeatable estimates of node traversal time as a function of congestion.  
An example SUMO environment for Node~10 is shown in Fig.~\ref{fig:sumo-node10}.

\textit{Base node-cost fitting.}  
The measured intersection delays are fitted to \textit{quartic} polynomials for each of the 24 intersections.  
This choice maintains consistency with the classical BPR edge-cost model, which also exhibits a fourth-order dependence on flow, enabling a unified treatment of link and node congestion. Table~\ref{tab:node_cost_coeffs} in Appendix \ref{app:sumo-details} lists all coefficients and mean absolute percentage errors (MAPE) for all nodes.  
Across all nodes, the quartic models achieve high accuracy (MAPE typically 1–2\%) while ensuring Assumption \ref{assm:cost} holds.

\textit{Incentive validation experiment.}  
To validate the analytical assumption that timestamp-based incentives act as additive constants in node costs, we conduct additional SUMO experiments on representative intersection, node 10.  
For one of the paths, we implement timestamp modifications directly at the RSU level.  
Specifically, we adjust vehicle request times by $\pm 10$ sec to represent the effect of (i) a positive incentive (\emph{timestamp advancement}) or (ii) a negative incentive (\emph{timestamp delay}).  
These offsets correspond to small scheduling shifts consistent with realistic RSU control capabilities.

In general, multiple vehicles following different paths may coexist on the same lane.  
To prevent unrealistic reordering or collision of timestamp schedules, we consider a local timestamp‐difference constraint:  
when delaying a vehicle, if the time gap between that vehicle and the one immediately behind it exceeds 10 s, the full 10 s delay is applied; otherwise, the delay is bounded by the available gap (e.g., 6 s if the next vehicle’s timestamp is only 7 s later).  
Conversely, when expediting a vehicle, if the timestamp gap to the vehicle in front is larger than 10 s, the full 10 s advancement is applied; if the gap is smaller, we advance the timestamp only up to that gap.  
This rule ensures that the applied incentives preserve the temporal ordering of the RSU’s scheduling policy.

\begin{figure}[t!]
    \centering   \includegraphics[width=.7\linewidth]{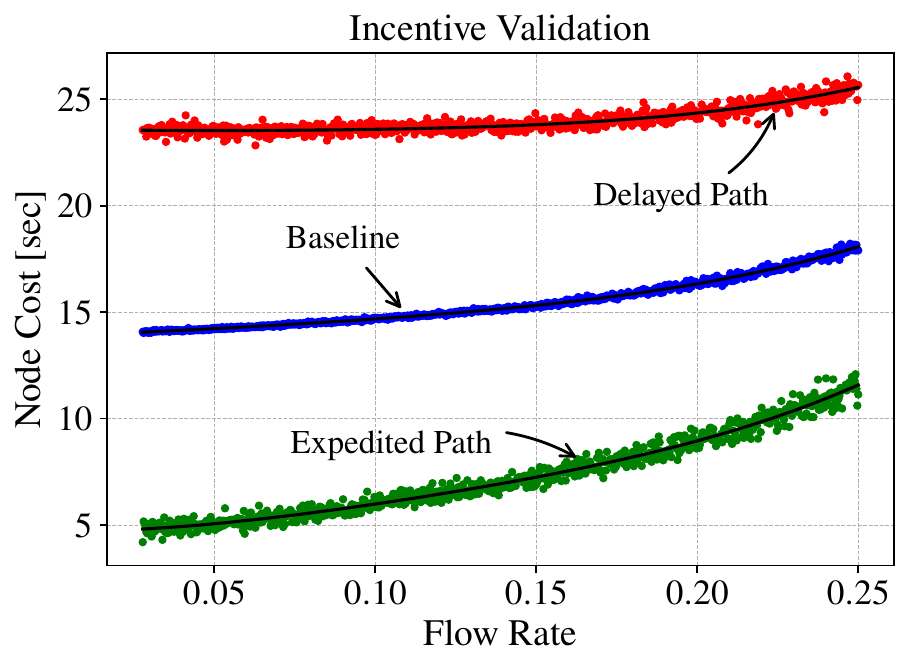}
    \caption{Average travel time (second) versus flow rate (vehicles per second) within the control zone and the flow rate for the simulated Node 10 using the FCFS protocol under no, positive and negative timestamp modifications for one of the paths.}
    \label{fig:node10-fit}
\end{figure}

Fig.~\ref{fig:node10-fit} shows the resulting empirical node‐cost fittings under positive, negative and, no incentives.  
Both perturbed curves remain approximately parallel to the baseline under no incentive case, confirming that timestamp adjustments primarily shift the node‐cost function vertically without altering its functional dependence on flow.  
This empirically verifies that $u_v^p$ acts as an additive offset in the cost model, consistent with the theoretical formulation $c_v^p(f_v, u_v^p) = c_v(f_v) + u_v^p$.

\textit{Bi-level optimization.}  
We integrate the SUMO-calibrated quartic node costs with the BPR edge costs. The lower-level user equilibrium (UE) is computed using the Frank–Wolfe algorithm, while the upper-level planner optimizes timestamp incentives via the Simultaneous Perturbation Stochastic Approximation (SPSA) method \citep{ref:Spall92}.  
Two bounded incentive regimes are examined to evaluate both theoretical scalability and practical feasibility:  
(i) a \emph{high-intervention} case with $u_v^p \in [0,2]$ min, allowing high timestamp offsets up to 2 min, and  
(ii) a \emph{low-intervention} case with $u_v^p \in [0,0.5]$ min (i.e., up to 30 sec), allowing minimal timestamp offsets. These dual settings allow assessment of how incentive magnitude influences network-level efficiency improvements while preserving physical realism and RSU timing constraints.  
All simulation and algorithmic parameters are detailed in Appendix~\ref{app:sumo-details}.

\begin{figure}[t!]
\centering 
\vspace{.065in}
\includegraphics[width=0.7\linewidth]{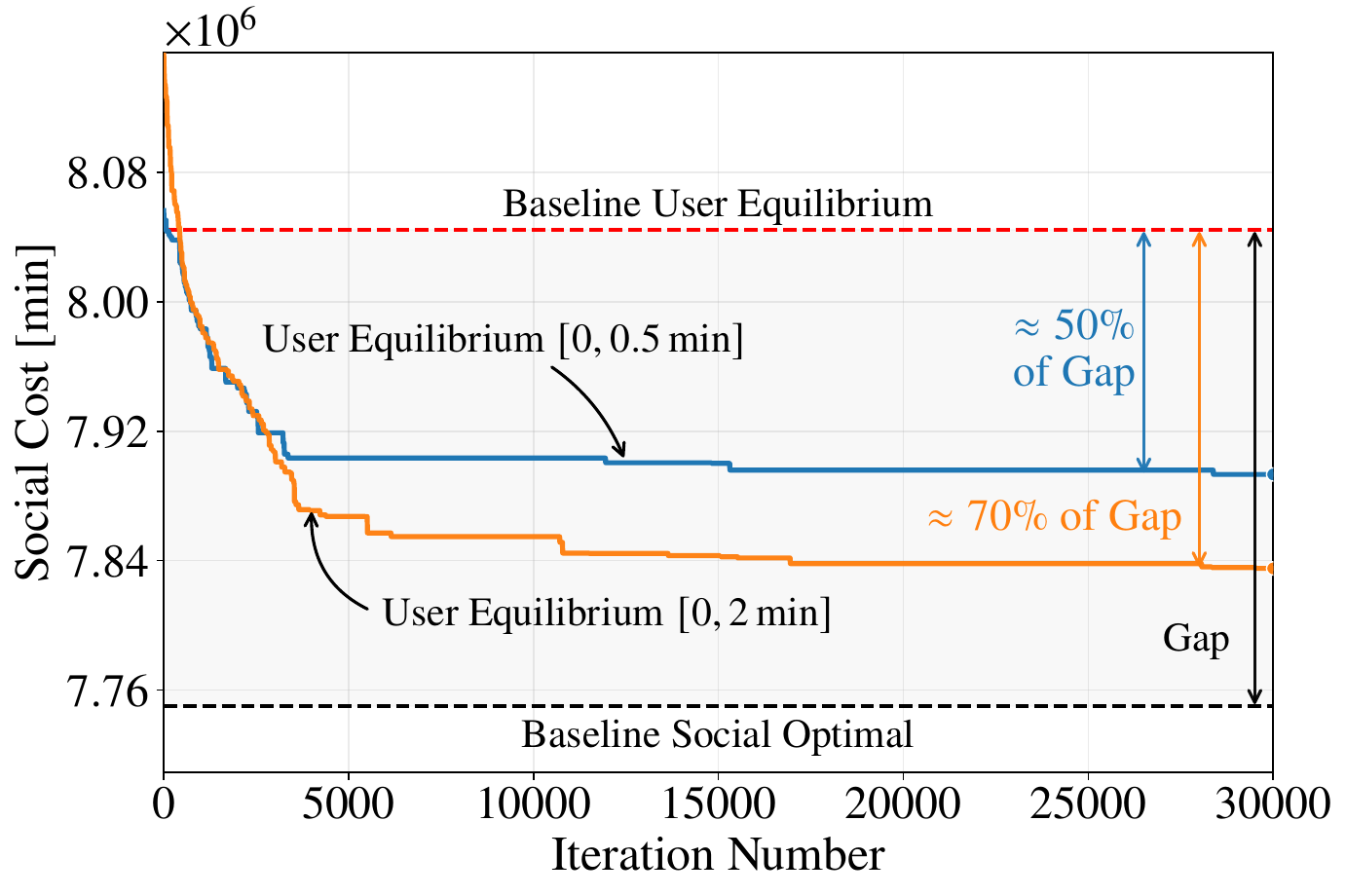} \caption{Convergence of computed social costs in SPSA with different incentive bounds.} \label{fig:resultsa} \end{figure}

\begin{figure}[t!]
  \centering
  \begin{subfigure}[b]{0.49\linewidth}
    \centering
    \includegraphics[width=\linewidth]{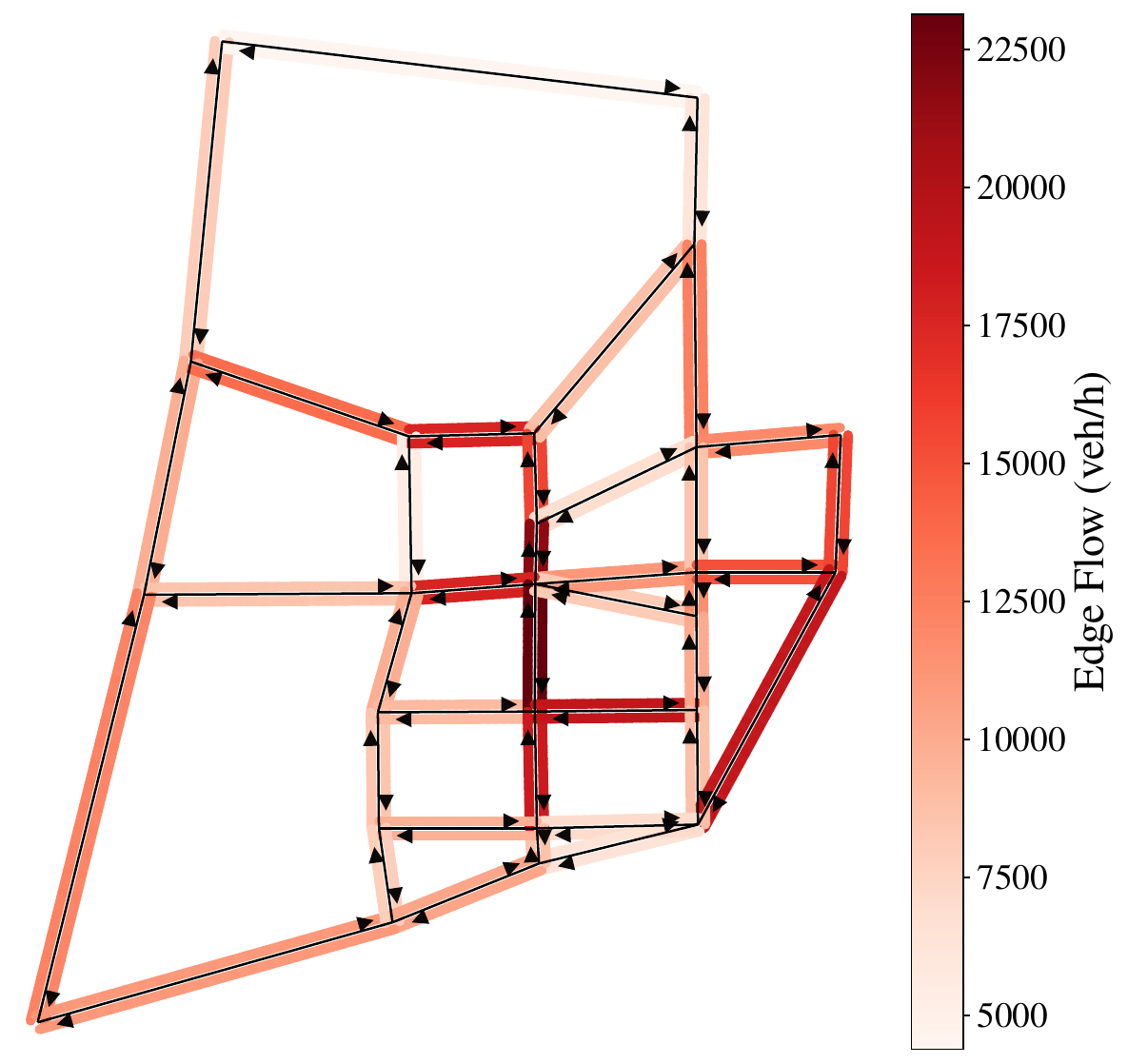}
    \caption{Baseline.}
    \label{fig:flow-heatmap-baseline}
  \end{subfigure}\hfill
  \begin{subfigure}[b]{0.49\linewidth}
    \centering
    \includegraphics[width=\linewidth]{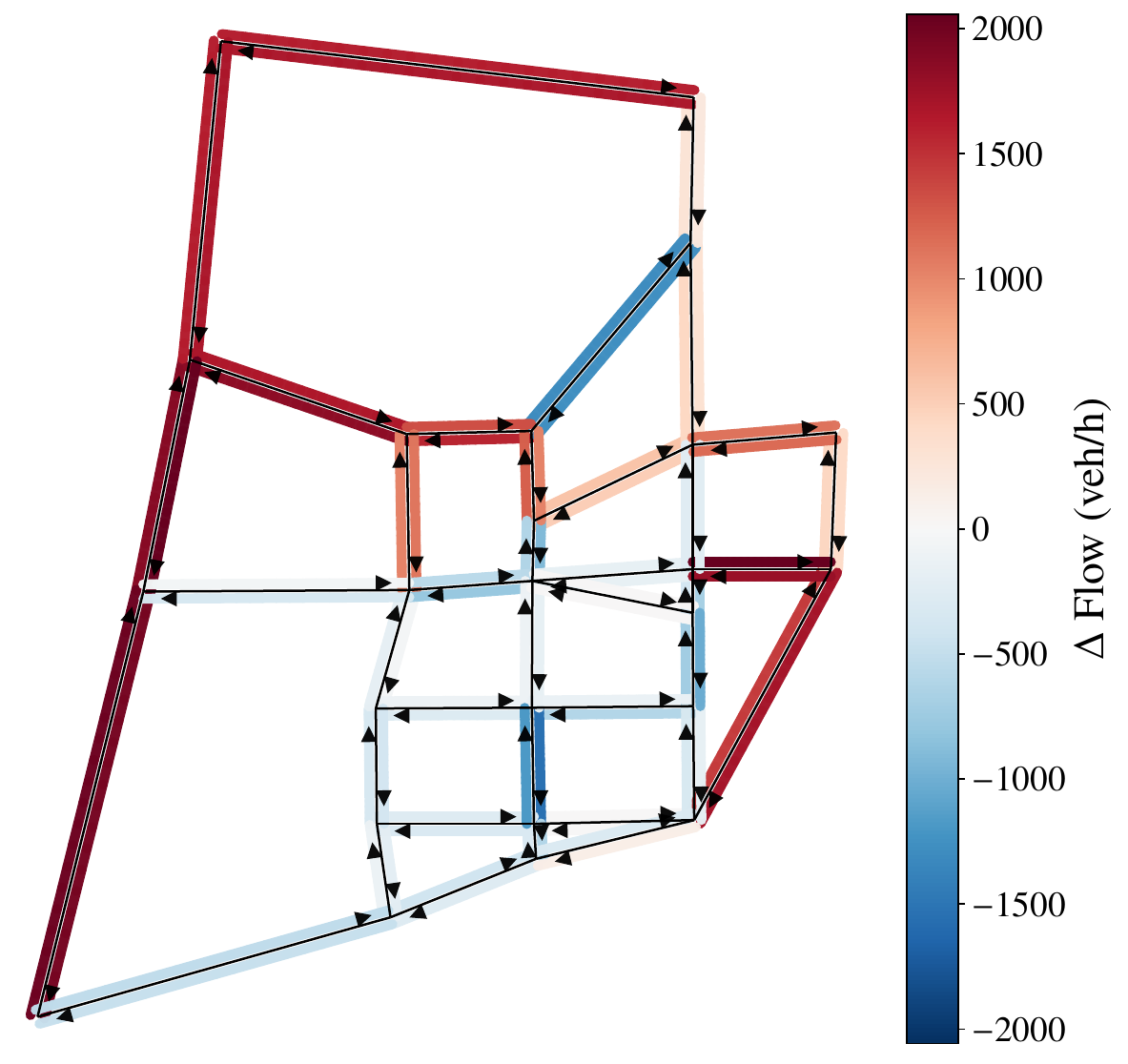}
    \caption{Difference.}
    \label{fig:flow-heatmap-diff}
  \end{subfigure}

  \caption{Comparison of network flow distributions in the Sioux Falls benchmark.  
  (a) Baseline user–equilibrium (UE) showing congestion concentration on central corridors.  
  (b) Flow redistribution after timestamp-based incentives, where red (and blue) edges depict increased (and decreased) flows.}
  \label{fig:resultsb}
\end{figure}

Fig.~\ref{fig:resultsa} illustrates the convergence of the social cost throughout SPSA iterations for both incentive ranges.  
The incentivized user equilibria—one with timestamp adjustments constrained to $[0,2]$\,min and another to $[0,0.5]$\,min—both progressively approach the system–optimal benchmark (black dashed), while the baseline user equilibrium (red dashed) remains substantially higher.  
The high-intervention case achieves lower social cost, whereas the low-intervention case incentives still recover most of the efficiency improvement, demonstrating effectiveness of the mechanism even under tight operational limits. In Fig.~\ref{fig:resultsa}, the starting points of the convergence plots are different due to their randomized initializations.

The baseline UE cost is $8.04 \times 10^6$ min, while the system–optimal benchmark achieves $7.75 \times 10^6$ min.  
Under the high-intervention case, the incentivized equilibrium reaches $7.84 \times 10^6$ min, closing $
71.1\%$
of the UE–to–SO optimality gap.  
When the incentive range is restricted to $[0,0.5]$ min, the incentivized equilibrium achieves $7.89 \times 10^6$ minutes, closing $51.3\%$ of the optimality gap.  
These results demonstrate that timestamp-based intersection incentives recover a substantial fraction of the attainable efficiency improvement even under tight operational limits.  
While broader incentive ranges yield stronger convergence and higher efficiency, the 30-sec case still captures a significant amount of the benefit, confirming the \textit{practical viability and robustness} of the proposed bilevel mechanism under realistic traffic-control constraints.

Fig.~\ref{fig:resultsb} compares baseline and incentivized flow distributions for the Sioux Falls network under the $[0,2]$-min incentive range.  
Fig.~\ref{fig:flow-heatmap-baseline} shows the baseline UE, where heavy congestion is concentrated on central corridors.  
Fig.~\ref{fig:flow-heatmap-diff} depicts the corresponding flow differences after incentives, with red links indicating increased usage and blue links decreased usage.  
These heatmaps illustrate how timestamp-based incentives alleviate pressure on over-saturated central routes, diverting part of the demand toward peripheral alternatives and thereby achieving a more balanced utilization across the network.

\section{Conclusion}\label{sec:conclusion}

This paper proposes a non-monetary control framework for improving traffic efficiency in urban networks via timestamp modifications at autonomous intersections. By decoupling incentive design from intersection scheduling in a two-layer architecture, localized control can guide self-interested drivers toward socially efficient flows without tolling. Our timestamp-modification-based prioritization scheme yields additive cost adjustments, establishing existence and essential uniqueness of equilibrium flows. The essential uniqueness is critical for scalability by mitigating equilibrium multiplicity and enabling a computationally tractable bilevel optimization framework. Numerical experiments demonstrate substantial improvements in efficiency in realistic traffic scenarios. These results highlight the potential of non-monetary, infrastructure-light control for scalable deployment in next-generation intelligent transportation and urban mobility systems.

\textit{Future Research Directions.}
Future work includes integrating monetary or information-based mechanisms to balance efficiency and equity, developing adaptive and real-time implementations (e.g., via reinforcement learning or online optimization), incorporating heterogeneous behavioral responses such as partial compliance, and extending the framework to mixed-autonomy settings with both human-driven and autonomous vehicles.

\appendix

\section{Proof of Proposition \ref{prop:potential}}\label{app:potential}

Consider a non-negative, continuous and non-decreasing function $c(\cdot)$. Then, the mapping $x\mapsto \int_0^x c(s)\,ds$ is continuously differentiable since $\frac{d}{dx}\int_0^x c(s)ds = c(x)$ by the fundamental theorem of calculus. Furthermore, the mapping is convex since
\begin{flalign*}
\int_x^y c(s)\,ds \geq c(x)(y-x) = \left(\frac{d}{dx}\int_0^x c(s)ds \right)(y-x).
\end{flalign*}
Therefore, Assumption \ref{assm:cost} yields that the mappings
\begin{flalign}
x\mapsto \int_0^x c_e(s)\,ds \quad \text{and} \quad y\mapsto \int_0^y c_v(s)\,ds
\end{flalign}
are convex and continuously differentiable. Moreover, $f\cdot u$ is linear. This yields that for a given incentive profile $u$, $\Phi(f,u)$ is a convex and continuously differentiable function of $f\in\mathcal{F}$.

For the equivalence, we prove the following two claims.

\begin{claim}\label{claim:forward}
For a given incentive profile $u$, if \( f \) is an optimal solution of \eqref{eq:potential_min}, then it is an equilibrium flow.
\end{claim}

\begin{proof}
Since the objective is convex and the constraints are affine, regularity conditions hold and the Karush--Kuhn--Tucker (KKT) conditions are both necessary and sufficient for optimality. Define the Lagrangian
\begin{flalign}
\mathcal{L}(f,\lambda,\mu)=\Phi(f, u)-\lambda\Bigl(\sum_{p\in \mathcal{P}}f_p-1\Bigr)-\sum_{p\in \mathcal{P}}f_p \cdot \mu_p.
\end{flalign}
Differentiating \(\Phi(f, u)\) with respect to \( f_p \),
we obtain
\begin{flalign}
\frac{\partial \Phi(f, u)}{\partial f_p}=\sum_{e\in p} c_e(f_e)+\sum_{v\in p} c_v^p(f_v, u_v^p) = c_p(f, u).
\end{flalign}
Thus, the stationarity condition reads
\begin{flalign}\label{eq:stationarity}
c_p(f, u)-\lambda-\mu_p=0,\quad \forall\, p\in \mathcal{P}.
\end{flalign}
Complementary slackness implies that if \( f_p>0 \), then \(\mu_p=0\) so that $c_p(f, u)=\lambda$, and if \( f_p=0 \), then \(\mu_p\ge0\) so that $c_p(f, u)=\lambda+\mu_p\ge \lambda$.
They correspond to the Wardrop conditions, as described in Definition \ref{def:equilibrium}. This proves that \( f \) is an equilibrium flow. 
\end{proof}

\begin{claim}\label{claim:backward}
For a given incentive profile $u$, if \( f \) is an equilibrium flow, then it is an optimal solution of \eqref{eq:potential_min}.
\end{claim}

\begin{proof}
Assume that \( f \) is an equilibrium flow with equilibrium cost \(\lambda\), i.e., $c_p(f, u)=\lambda$ if $f_p>0$ and $c_p(f, u)\geq \lambda$ if $f_p=0$. Define multipliers by setting \(\mu_p=0\) for \( f_p>0 \) and \(\mu_p=c_p(f, u)-\lambda\ge0\) for \( f_p=0 \). Then the KKT stationarity condition \eqref{eq:stationarity} is satisfied. Together with primal feasibility, dual feasibility, and complementary slackness, the KKT conditions hold. Since these conditions are sufficient for optimality in convex problems, it follows that \( f \) is an optimal solution. 
\end{proof}

Claims \ref{claim:forward} and \ref{claim:backward} complete the proof of Proposition~\ref{prop:potential}.

\section{Proof of Theorem \ref{thm:existence}}\label{app:existence}

 The flow conservation condition $\mathcal{F}$, as described in \eqref{eq:feasible}, is a compact subset of \(\mathbb{R}^{|\mathcal{P}|}\) and for a given incentive profile $u$, \(\Phi(f, u)\) is continuous on \(\mathcal{F}\) by the continuity of \(c_e(\cdot)\) and \(c_v(\cdot)\). Therefore, the Weierstrass Extreme Value Theorem ensures that \(\Phi(\cdot)\) attains its minimum over \(\mathcal{F}\) for any given $u$. To establish essential uniqueness, let \( f \) and \(\tilde{f}\) be two equilibrium flows associated with $u$. Since both $f$ and $\tilde{f}$ minimize \(\Phi(f)\) by Proposition~\ref{prop:potential} and \(\Phi\) is convex, we have
\begin{flalign}
\Phi(f,u)
&\leq \alpha\Phi(f,u)+(1-\alpha)\Phi(\tilde{f},u)
=\Phi(f,u)
\end{flalign}
for all $\alpha\in[0,1]$. The equality for all \(\alpha\) ensures that the integrals \(\int_0^{f_e} c_e(x)\,dx\) and \(\int_0^{f_v} c_v(y)\,dy\) are linear between \( f_e, \tilde{f}_e \) and \( f_v, \tilde{f}_v \), respectively. This implies that their derivatives—corresponding to the marginal costs—are equal.
Then, \eqref{eq:pathcost} yields that
\be\label{eq:thesamepathcost}
c_p(f,u) = c_p(\tilde{f},u)\quad\forall p\in \mathcal{P}.
\ee
Note that the equilibrium condition implies that each path with positive flow under $f$ or $\tilde{f}$ has \textit{the same and the smallest} path costs, i.e., for some $\lambda$ and $\tilde{\lambda}$, we have 
\begin{subequations}
\begin{flalign}
\lambda = c_p(f,u) \leq c_{\bar{p}}(f,u)\quad\forall p,\bar{p}\mbox{ such that }f_p>0,\\
\tilde{\lambda} = c_p(\tilde{f},u) \leq c_{\bar{p}}(\tilde{f},u)\quad\forall p,\bar{p}\mbox{ such that }\tilde{f}_p>0.
\end{flalign}
\end{subequations}
Then, \eqref{eq:thesamepathcost} yields that $\lambda=\tilde{\lambda}$. The proof is completed since $C(f,u) = \lambda$ and $C(\tilde{f},u) = \tilde{\lambda}$ by \eqref{eq:socialcost} and $f\in\mathcal{F}$.

\section{Experimental Setup and Reproducibility}
\label{app:sumo-details}

\subsection{SUMO Scenario Design}
Each Sioux Falls intersection is modeled as an independent SUMO scenario that preserves its local geometry, signal timing, and incoming/outgoing lanes.  
Vehicles are generated according to a Poisson process with flow rate $\lambda = \frac{N}{T}$, where $N$ is the expected number of departures (ranging from 0 to 900) during the simulation horizon $T=3600\,\text{s}$.  
Vehicle routes are selected uniformly among feasible turn movements.  
For each vehicle, traversal time is measured from 25\,m upstream of the stop line representing the RSU’s control zone until complete exit from the intersection.

\begin{table}[t!]
\centering
\caption{Fitted quartic node cost coefficients from SUMO simulations for the 24 Sioux Falls intersections. 
Quartic models follow $c_v(x) = a_4 x^4 + a_3 x^3 + a_2 x^2 + a_1 x + a_0$. 
Columns list coefficients ($a_0$–$a_4$) and fitting error (MAPE, \%). Rows correspond to the 24 nodes.}
\label{tab:node_cost_coeffs}
\setlength{\tabcolsep}{2pt}
\renewcommand{\arraystretch}{0.9}
\begin{tabular}{rcccccc}
\toprule
     & \multicolumn{5}{c}{\textbf{Quartic: } $y = a_0 + a_1 N + a_2 N^2 + a_3 N^3 + a_4 N^4$} & \\
\midrule
$v$ & {$a_4$} & {$a_3$} & {$a_2$} & {$a_1$} & {$a_0$} & {\%} \\
\midrule
1  & $-1.14\!\times\!10^{-14}$ & $3.68\!\times\!10^{-10}$ & $-2.54\!\times\!10^{-7}$ & $1.91\!\times\!10^{-4}$ & $13.99$ & $0.43$ \\
2  & $-3.91\!\times\!10^{-14}$ & $4.92\!\times\!10^{-10}$ & $-4.53\!\times\!10^{-7}$ & $3.36\!\times\!10^{-4}$ & $13.94$ & $0.42$ \\
3  & $-3.98\!\times\!10^{-12}$ & $1.96\!\times\!10^{-8}$ & $-1.56\!\times\!10^{-5}$ & $5.38\!\times\!10^{-3}$ & $12.96$ & $1.64$ \\
4  & $-2.81\!\times\!10^{-12}$ & $1.74\!\times\!10^{-8}$ & $-1.46\!\times\!10^{-5}$ & $5.41\!\times\!10^{-3}$ & $12.91$ & $1.74$ \\
5  & $-4.78\!\times\!10^{-12}$ & $2.21\!\times\!10^{-8}$ & $-1.81\!\times\!10^{-5}$ & $6.23\!\times\!10^{-3}$ & $12.90$ & $1.69$ \\
6  & $-4.48\!\times\!10^{-12}$ & $2.09\!\times\!10^{-8}$ & $-1.68\!\times\!10^{-5}$ & $5.82\!\times\!10^{-3}$ & $12.89$ & $1.68$ \\
7  & $-5.75\!\times\!10^{-14}$ & $5.56\!\times\!10^{-10}$ & $-5.20\!\times\!10^{-7}$ & $3.62\!\times\!10^{-4}$ & $13.93$ & $0.44$ \\
8  & $6.63\!\times\!10^{-13}$ & $1.02\!\times\!10^{-8}$ & $-1.06\!\times\!10^{-5}$ & $5.08\!\times\!10^{-3}$ & $13.16$ & $1.94$ \\
9  & $-3.35\!\times\!10^{-12}$ & $1.80\!\times\!10^{-8}$ & $-1.43\!\times\!10^{-5}$ & $5.05\!\times\!10^{-3}$ & $12.98$ & $1.67$ \\
10 & $1.40\!\times\!10^{-12}$ & $1.26\!\times\!10^{-8}$ & $-1.45\!\times\!10^{-5}$ & $6.93\!\times\!10^{-3}$ & $13.54$ & $2.21$ \\
11 & $1.84\!\times\!10^{-12}$ & $8.22\!\times\!10^{-9}$ & $-9.92\!\times\!10^{-6}$ & $5.16\!\times\!10^{-3}$ & $13.15$ & $1.83$ \\
12 & $-5.06\!\times\!10^{-12}$ & $2.23\!\times\!10^{-8}$ & $-1.78\!\times\!10^{-5}$ & $6.08\!\times\!10^{-3}$ & $12.89$ & $1.73$ \\
13 & $6.28\!\times\!10^{-15}$ & $3.10\!\times\!10^{-10}$ & $-2.11\!\times\!10^{-7}$ & $2.16\!\times\!10^{-4}$ & $13.95$ & $0.45$ \\
14 & $-3.83\!\times\!10^{-12}$ & $1.91\!\times\!10^{-8}$ & $-1.53\!\times\!10^{-5}$ & $5.47\!\times\!10^{-3}$ & $12.92$ & $1.59$ \\
15 & $4.68\!\times\!10^{-13}$ & $1.12\!\times\!10^{-8}$ & $-1.17\!\times\!10^{-5}$ & $5.39\!\times\!10^{-3}$ & $13.16$ & $1.90$ \\
16 & $-3.58\!\times\!10^{-12}$ & $2.24\!\times\!10^{-8}$ & $-2.15\!\times\!10^{-5}$ & $8.31\!\times\!10^{-3}$ & $12.95$ & $2.12$ \\
17 & $-5.66\!\times\!10^{-12}$ & $2.42\!\times\!10^{-8}$ & $-1.95\!\times\!10^{-5}$ & $6.74\!\times\!10^{-3}$ & $12.87$ & $1.75$ \\
18 & $-5.30\!\times\!10^{-12}$ & $2.39\!\times\!10^{-8}$ & $-1.92\!\times\!10^{-5}$ & $6.50\!\times\!10^{-3}$ & $13.05$ & $1.74$ \\
19 & $-2.05\!\times\!10^{-12}$ & $1.53\!\times\!10^{-8}$ & $-1.27\!\times\!10^{-5}$ & $4.82\!\times\!10^{-3}$ & $12.95$ & $1.63$ \\
20 & $-4.90\!\times\!10^{-12}$ & $3.15\!\times\!10^{-8}$ & $-2.75\!\times\!10^{-5}$ & $9.61\!\times\!10^{-3}$ & $14.50$ & $2.27$ \\
21 & $-5.02\!\times\!10^{-12}$ & $2.29\!\times\!10^{-8}$ & $-1.90\!\times\!10^{-5}$ & $6.70\!\times\!10^{-3}$ & $12.81$ & $1.76$ \\
22 & $3.01\!\times\!10^{-12}$ & $4.78\!\times\!10^{-9}$ & $-6.83\!\times\!10^{-6}$ & $4.19\!\times\!10^{-3}$ & $13.39$ & $1.88$ \\
23 & $-3.66\!\times\!10^{-12}$ & $1.89\!\times\!10^{-8}$ & $-1.53\!\times\!10^{-5}$ & $5.44\!\times\!10^{-3}$ & $12.93$ & $1.68$ \\
24 & $-3.44\!\times\!10^{-12}$ & $1.88\!\times\!10^{-8}$ & $-1.55\!\times\!10^{-5}$ & $5.49\!\times\!10^{-3}$ & $12.94$ & $1.73$ \\
\bottomrule
\end{tabular}
\end{table}

Measured travel times are averaged across 20 rollouts. Each node’s empirical cost data are then fitted to a quartic model using least squares regression. Typical MAPE fall within 1–2\%.  Representative fits for Node~10 are shown in Fig.~\ref{fig:node10-fit} and the computed quartic function coefficients are tabulated in Table \ref{tab:node_cost_coeffs} for $24$ intersections, illustrated in Fig. \ref{fig:sioux-overview}.

\subsection{Bilevel Optimization Configuration}
For the lower-level UE computations, we use the Frank–Wolfe algorithm with tolerance $10^{-6}$ and step size $2/(k+2)$.  
We solve the upper-level bilevel optimization (planner problem) via the SPSA algorithm using $30,000$ iterations with parameters
$a = 0.1$, $\alpha = 0.40$, $c = 0.4$, $\gamma = 0.03$, and $A = 1200$. Each iteration requires two equilibrium evaluations, resulting in a total of $60,000$ UE solves.  
This setup provides stable convergence while maintaining computational tractability. Lastly, all experiments were conducted on a workstation equipped with an
Intel\textsuperscript{\textregistered} Xeon\textsuperscript{\textregistered} w7-3455 CPU  
and an NVIDIA RTX~A4000 GPU (16\,GB VRAM).  
The bilevel optimization experiments required approximately two days of runtime.

\end{spacing}
\begin{spacing}{1}
\bibliographystyle{plainnat}
\bibliography{myref}
\end{spacing}

\end{document}